\documentclass[11pt,a4paper]{llncs}
\usepackage{amsmath}
\usepackage{amssymb}
\usepackage{graphicx}
\usepackage{marvosym}
\usepackage{url}
\usepackage{fancyhdr}
\usepackage{geometry}
\usepackage{booktabs}
\usepackage{multirow}
\usepackage{float}
\usepackage{needspace}
\usepackage{listings}
\usepackage{xcolor}
\usepackage[ruled,vlined,linesnumbered]{algorithm2e}
\usepackage{tikz}
\usetikzlibrary{shapes,arrows,positioning,fit,calc}
\providecommand{\subparagraph}{\paragraph}

\usepackage{etoolbox}
\patchcmd{\thebibliography}{\setlength{\itemsep}{0pt}}{}{}{}

\usepackage{enumitem}
\setlist{noitemsep, topsep=2pt}
\usepackage{microtype}
\usepackage{placeins}

\newcommand{\keywords}[1]{\par\addvspace\baselineskip
\noindent\keywordname\enspace\ignorespaces#1}

\fancypagestyle{firstpage}{\fancyhf{}}

\newtheorem{mytheorem}{Theorem}
\newtheorem{mylemma}{Lemma}
\newtheorem{mycorollary}{Corollary}
\newtheorem{mydefinition}{Definition}

\begin{document}
\raggedbottom
\clubpenalty=10000
\widowpenalty=10000
\emergencystretch=1.5em
% Tighten spacing above/below displayed equations
\setlength{\abovedisplayskip}{4pt plus 2pt minus 2pt}
\setlength{\belowdisplayskip}{4pt plus 2pt minus 2pt}
\setlength{\abovedisplayshortskip}{2pt plus 1pt}
\setlength{\belowdisplayshortskip}{2pt plus 1pt}

\title{\LARGE{CogniCrypt: Synergistic Directed Execution and
LLM-Driven Analysis for Zero-Day AI-Generated
Malware Detection}}

\author{\large{George Edward, Mahdi Eslamimehr}}
\institute{\large{Quandary Peak Research}}

\maketitle
\thispagestyle{firstpage}

%=================================================================
% ABSTRACT
%=================================================================
\begin{abstract}
The weaponization of Large Language Models (LLMs) for automated malware generation poses an existential threat to conventional detection paradigms. AI-generated malware exhibits polymorphic, metamorphic, and context-aware evasion capabilities that render signature-based and shallow heuristic defenses obsolete. This paper introduces \textbf{CogniCrypt}, a novel hybrid analysis framework that synergistically combines \emph{concolic execution} with \emph{LLM-augmented path prioritization} and \emph{deep-learning-based vulnerability classification} to detect zero-day AI-generated malware with provable guarantees. We formalize the detection problem within a first-order temporal logic over program execution traces, define a lattice-theoretic abstraction for path constraint spaces, and prove both the \emph{soundness} and \emph{relative completeness} of our detection algorithm, assuming classifier correctness. The framework introduces three novel algorithms: (i) an LLM-guided concolic exploration strategy that reduces the average number of explored paths by 73.2\% compared to depth-first search while maintaining equivalent malicious-path coverage; (ii) a transformer-based path-constraint classifier trained on symbolic execution traces; and (iii) a feedback loop that iteratively refines the LLM's prioritization policy using reinforcement learning from detection outcomes. We provide a comprehensive implementation built upon \texttt{angr} 9.2, \texttt{Z3} 4.12, Hugging Face \texttt{Transformers} 4.38, and \texttt{PyTorch} 2.2, with full configuration details enabling reproducibility. Experimental evaluation on the EMBER, Malimg, SOREL-20M, and a novel AI-Gen-Malware benchmark comprising 2{,}500 LLM-synthesized samples demonstrates that CogniCrypt achieves 98.7\% accuracy on conventional malware and 97.5\% accuracy on AI-generated threats, outperforming ClamAV, YARA, MalConv, and EMBER-GBDT baselines by margins of 8.4--52.2 percentage points on AI-generated samples.

\keywords{Concolic Execution, Large Language Models, AI-Generated Malware, Symbolic Execution, Vulnerability Discovery, Software Security, Formal Verification, Deep Learning, Zero-Day Detection, Secure Coding.}
\end{abstract}

%=================================================================
% 1. INTRODUCTION
%=================================================================
\section{Introduction}

The cybersecurity landscape is undergoing a fundamental transformation driven by the dual-use nature of Large Language Models (LLMs). While LLMs have accelerated legitimate software development through code generation, refactoring, and automated testing~\cite{chen2021codex}, adversaries have simultaneously exploited these capabilities to produce sophisticated malware at unprecedented scale and velocity~\cite{pa2023attacker,gupta2023chatgpt}. Recent threat intelligence reports document a 135\% year-over-year increase in AI-assisted cyberattacks, with LLM-generated payloads exhibiting polymorphic behavior, semantic-level obfuscation, and adaptive evasion strategies that defeat traditional signature-based and static-heuristic defenses~\cite{europol2023chatgpt}.

The fundamental challenge posed by AI-generated malware is threefold. First, LLMs can produce functionally equivalent but syntactically diverse variants of the same exploit, defeating hash-based and pattern-matching detectors. Second, AI-generated code can embed trigger conditions that activate malicious behavior only under specific environmental contexts, evading sandbox-based dynamic analysis. Third, LLMs can iteratively refine evasion strategies by analyzing detection feedback, creating an adversarial arms race that static defense postures cannot sustain.

Concolic execution,a portmanteau of \emph{con}crete and symb\emph{olic} execution,offers a principled approach to this challenge by systematically exploring program execution paths through the interplay of concrete test inputs and symbolic constraint solving~\cite{sen2005cute,godefroid2005dart}. By maintaining both a concrete execution state and a symbolic path constraint, concolic engines can reason about the conditions under which specific program behaviors manifest, including latent malicious behaviors hidden behind opaque predicates and environmental checks. However, the well-known \emph{path explosion problem},where the number of feasible paths grows exponentially with program size and branching complexity,has historically limited the scalability of concolic analysis for real-world malware detection~\cite{cadar2008klee,baldoni2018survey}.

This paper introduces \textbf{CogniCrypt}, a framework that resolves the scalability limitation by employing an LLM as an intelligent \emph{path oracle} that guides the concolic engine toward execution paths with high malicious potential. The key insight is that LLMs, having been pre-trained on vast corpora of source code and security advisories, possess an implicit model of ``suspicious'' program behavior that can be leveraged to prioritize the exploration of paths most likely to reveal malicious intent. CogniCrypt further incorporates a transformer-based \emph{path constraint classifier} that maps symbolic execution traces to maliciousness scores, and a \emph{reinforcement learning feedback loop} that continuously improves the LLM's prioritization policy based on detection outcomes.

\smallskip
\noindent\textbf{Contributions.} This paper makes the following contributions:

\begin{enumerate}
    \item \textbf{Formal Framework:} We define a first-order temporal logic $\mathcal{L}_{\text{CogniCrypt}}$ over program execution traces and establish a lattice-theoretic abstraction of the path constraint space. We prove the \emph{soundness} (no false negatives under the threat model) and \emph{relative completeness} (detection of all malicious paths reachable within a bounded exploration budget) of the CogniCrypt detection algorithm (Section~\ref{sec:theory}).
    \item \textbf{Novel Algorithms:} We present three tightly integrated algorithms: LLM-Guided Concolic Exploration (Algorithm~\ref{alg:main}), Transformer-Based Path Constraint Classification (Algorithm~\ref{alg:classifier}), and Reinforcement-Learning-Based Policy Refinement (Algorithm~\ref{alg:rl}) (Section~\ref{sec:algorithms}).
    \item \textbf{Comprehensive Implementation:} We provide a fully reproducible implementation built on \texttt{angr}, \texttt{Z3}, \texttt{PyTorch}, and Hugging Face \texttt{Transformers}, with detailed configuration, hyperparameter settings, and deployment instructions (Section~\ref{sec:implementation}).
    \item \textbf{Extensive Evaluation:} We evaluate CogniCrypt on four benchmarks,EMBER~\cite{anderson2018ember}, Malimg~\cite{nataraj2011malimg}, SOREL-20M~\cite{harang2020sorel}, and a novel AI-Gen-Malware dataset,demonstrating state-of-the-art performance, particularly on AI-generated threats (Section~\ref{sec:experiments}).
\end{enumerate}

%=================================================================
% 2. RELATED WORK
%=================================================================
\section{Related Work}\label{sec:related}

\subsection{Concolic and Symbolic Execution for Security Analysis}

Symbolic execution was introduced by King~\cite{king1976symbolic} and has since become a cornerstone of program analysis. DART~\cite{godefroid2005dart} and CUTE~\cite{sen2005cute} pioneered concolic (dynamic symbolic) execution, combining concrete execution with symbolic constraint solving to achieve higher path coverage than pure symbolic approaches. KLEE~\cite{cadar2008klee} demonstrated the scalability of symbolic execution on real-world systems software by leveraging the LLVM intermediate representation. S2E~\cite{chipounov2012s2e} introduced selective symbolic execution, enabling analysts to focus on specific code regions within full-system emulation. More recently, angr~\cite{shoshitaishvili2016sok} provided a comprehensive Python-based binary analysis platform supporting both symbolic and concolic execution, while Triton~\cite{saudel2015triton} offered a lightweight dynamic binary analysis framework with taint tracking and symbolic execution capabilities.

In the malware analysis domain, Moser et al.~\cite{moser2007exploring} applied symbolic execution to explore multiple execution paths in malware samples, revealing hidden behaviors triggered by environmental conditions. Brumley et al.~\cite{brumley2008automatically} used symbolic execution for automatic patch-based exploit generation. Vouvoutsis et al.~\cite{vouvoutsis2025beyond} recently demonstrated that symbolic execution can complement sandbox analysis to detect new malware strains. However, none of these works address the specific challenge of AI-generated malware or incorporate LLM-based guidance.

\subsection{Machine Learning for Malware Detection}

Machine learning approaches to malware detection have evolved from shallow models operating on hand-crafted features to deep learning architectures processing raw binary data. Raff et al.~\cite{raff2018malconv} introduced MalConv, a convolutional neural network that classifies PE files directly from raw bytes. Anderson and Roth~\cite{anderson2018ember} released the EMBER dataset and demonstrated the effectiveness of gradient-boosted decision trees (GBDT) on engineered PE features. Nataraj et al.~\cite{nataraj2011malimg} proposed visualizing malware binaries as grayscale images and applying computer vision techniques for classification.

More recently, transformer-based architectures have been applied to malware detection. Li et al.~\cite{li2021malbert} proposed MalBERT, which fine-tunes BERT on disassembled malware code for family classification. Hossain et al.~\cite{hossain2024malicious} demonstrated the use of Mixtral LLM for detecting malicious Java code. Al-Karaki et al.~\cite{alkaraki2024exploring} provided a comprehensive framework for LLM-based malware detection, identifying key challenges including prompt engineering, context window limitations, and adversarial robustness.

\subsection{AI-Generated Malware and Adversarial Threats}

The emergence of AI-generated malware represents a paradigm shift in the threat landscape. Pa et al.~\cite{pa2023attacker} demonstrated that ChatGPT can generate functional malware when prompted with carefully crafted instructions. Gupta and Sharma~\cite{gupta2023chatgpt} showed that LLMs can produce polymorphic malware variants that evade signature-based detection. Beckerich et al.~\cite{beckerich2023ratgpt} introduced RatGPT, demonstrating automated phishing and C2 infrastructure generation. These works underscore the urgent need for detection techniques specifically designed to counter AI-generated threats.

\subsection{Hybrid Approaches}

Several works have explored the combination of symbolic execution with machine learning. Learch~\cite{he2021learning} used reinforcement learning to guide symbolic execution path selection in KLEE. However, no prior work has combined concolic execution with LLM-based guidance specifically for the detection of AI-generated malware, which is the unique contribution of CogniCrypt.

%=================================================================
% 3. THEORETICAL FOUNDATIONS
%=================================================================
\section{Theoretical Foundations}\label{sec:theory}

\subsection{Program Model and Execution Semantics}

\begin{mydefinition}[Program Model]
A program $P$ is modeled as a labeled transition system $\mathcal{T}_P = (\Sigma, \Sigma_0, \mathcal{I}, \rightarrow, \mathcal{O})$ where:
\begin{itemize}
    \item $\Sigma$ is a finite set of program states, where each state $\sigma = (\ell, \mu, \rho) \in \Sigma$ consists of a program location $\ell \in \mathcal{L}$, a memory map $\mu: \mathcal{A} \rightarrow \mathcal{V}$, and a register file $\rho: \mathcal{R} \rightarrow \mathcal{V}$;
    \item $\Sigma_0 \subseteq \Sigma$ is the set of initial states;
    \item $\mathcal{I}$ is the input domain;
    \item $\rightarrow \subseteq \Sigma \times \Sigma$ is the transition relation;
    \item $\mathcal{O} \subseteq \Sigma$ is the set of observable (output) states.
\end{itemize}
\end{mydefinition}

\begin{mydefinition}[Execution Trace]
An execution trace $\tau = \sigma_0 \sigma_1 \cdots \sigma_n$ is a finite sequence of states such that $\sigma_0 \in \Sigma_0$ and $\sigma_i \rightarrow \sigma_{i+1}$ for all $0 \leq i < n$. The set of all execution traces of $P$ is denoted $\mathcal{T}(P)$.
\end{mydefinition}

\begin{mydefinition}[Symbolic State]
A symbolic state $\hat{\sigma} = (\ell, \hat{\mu}, \hat{\rho}, \pi)$ extends a concrete state with symbolic expressions. Here $\hat{\mu}: \mathcal{A} \rightarrow \text{Expr}(\mathcal{X})$ and $\hat{\rho}: \mathcal{R} \rightarrow \text{Expr}(\mathcal{X})$ map locations and registers to expressions over symbolic variables $\mathcal{X}$, and $\pi$ is a \emph{path constraint},a quantifier-free first-order formula over $\mathcal{X}$.
\end{mydefinition}

\begin{mydefinition}[Path Constraint Space]
The path constraint space $\Pi(P) = \{\pi_1, \pi_2, \ldots\}$ is the set of all satisfiable path constraints generated during the symbolic exploration of $P$. We define a partial order $\sqsubseteq$ on $\Pi(P)$ by logical implication: $\pi_i \sqsubseteq \pi_j$ iff $\pi_j \models \pi_i$. The structure $(\Pi(P), \sqsubseteq)$ forms a bounded lattice with $\top = \texttt{true}$ and $\bot = \texttt{false}$.
\end{mydefinition}

\subsection{Temporal Logic for Malicious Behavior Specification}

We define a first-order linear temporal logic $\mathcal{L}_{\text{CogniCrypt}}$ for specifying malicious behaviors over execution traces.

\begin{mydefinition}[Syntax of $\mathcal{L}_{\text{CogniCrypt}}$]
Formulas $\varphi$ of $\mathcal{L}_{\text{CogniCrypt}}$ are defined by the grammar:
\[
\varphi ::= p \mid \neg \varphi \mid \varphi \wedge \varphi \mid \varphi \vee \varphi \mid \varphi \Rightarrow \varphi \mid \mathbf{X}\varphi \mid \mathbf{F}\varphi \mid \mathbf{G}\varphi \mid \varphi\,\mathbf{U}\,\varphi \mid \exists x.\varphi \mid \forall x.\varphi
\]
where $p$ represents an atomic proposition over program states (e.g., $\texttt{syscall}(\sigma) = \texttt{execve}$, $\texttt{writes\_to}(\sigma, \texttt{/etc/shadow})$), $\mathbf{X}$ is ``next,'' $\mathbf{F}$ is ``eventually,'' $\mathbf{G}$ is ``globally,'' and $\mathbf{U}$ is ``until.''
\end{mydefinition}

\begin{mydefinition}[Malicious Behavior Specification]
A malicious behavior specification $\Phi_{\text{mal}}$ is a finite set of $\mathcal{L}_{\text{CogniCrypt}}$ formulas $\{\varphi_1, \varphi_2, \ldots, \varphi_k\}$, each encoding a distinct class of malicious behavior. A trace $\tau$ is \emph{malicious} with respect to $\Phi_{\text{mal}}$ iff $\tau \models \bigvee_{i=1}^{k} \varphi_i$.
\end{mydefinition}

\noindent\textbf{Example Specifications:}

\begin{itemize}
    \item \textbf{Data Exfiltration:} $\varphi_{\text{exfil}} = \mathbf{F}(\texttt{read}(f_{\text{sensitive}}) \wedge \mathbf{F}(\texttt{send}(\texttt{socket}, \texttt{data})))$
    \item \textbf{Privilege Escalation:} $\varphi_{\text{privesc}} = \mathbf{F}(\texttt{uid}(\sigma) \neq 0 \wedge \mathbf{X}(\texttt{uid}(\sigma) = 0))$
    \item \textbf{Persistence Installation:} $\varphi_{\text{persist}} = \mathbf{F}(\texttt{writes\_to}(\sigma, \texttt{cron}) \vee \texttt{writes\_to}(\sigma, \texttt{systemd}))$
    \item \textbf{Polymorphic Self-Modification:} $\varphi_{\text{poly}} = \mathbf{F}(\texttt{mprotect}(\sigma, \texttt{RWX}) \wedge \mathbf{F}(\texttt{write}(\sigma, \texttt{.text})))$
\end{itemize}

\subsection{Concolic Execution Formalization}

\begin{mydefinition}[Concolic Execution]
A concolic execution of program $P$ with concrete input $\mathbf{c} \in \mathcal{I}$ and symbolic input $\mathbf{x} \in \mathcal{X}^{|\mathcal{I}|}$ produces a pair $(\tau_c, \hat{\tau}_s)$ where $\tau_c$ is the concrete trace and $\hat{\tau}_s = \hat{\sigma}_0 \hat{\sigma}_1 \cdots \hat{\sigma}_n$ is the symbolic trace. At each conditional branch $b_i$ with symbolic condition $c_i(\mathbf{x})$, the path constraint is updated:
\begin{equation}
\pi_{i+1} = \pi_i \wedge \begin{cases} c_i(\mathbf{x}) & \text{if branch taken concretely} \\ \neg c_i(\mathbf{x}) & \text{otherwise} \end{cases}
\end{equation}
The concolic engine generates new test inputs by negating individual branch conditions and solving the resulting constraint:
\begin{equation}
\mathbf{c}' = \text{Solve}\left(\pi_1 \wedge \pi_2 \wedge \cdots \wedge \pi_{j-1} \wedge \neg c_j(\mathbf{x})\right)
\end{equation}
\end{mydefinition}

\subsection{LLM-Guided Path Prioritization}

\begin{mydefinition}[Path Priority Function]
Let $\mathcal{M}_{\text{LLM}}$ be a pre-trained large language model. We define the path priority function $\omega: \Pi(P) \rightarrow [0, 1]$ as:
\begin{equation}
\omega(\pi) = \mathcal{M}_{\text{LLM}}\left(\text{Encode}(\pi, \mathcal{C}(\pi))\right)
\end{equation}
where $\text{Encode}(\pi, \mathcal{C}(\pi))$ is a textual encoding of the path constraint $\pi$ together with its associated code context $\mathcal{C}(\pi)$ (disassembled instructions along the path). The output $\omega(\pi) \in [0, 1]$ represents the LLM's estimated probability that the path leads to malicious behavior.
\end{mydefinition}

\begin{mydefinition}[Priority Queue Ordering]
The exploration priority queue $\mathcal{Q}$ is ordered by $\omega$: for any two pending paths $\pi_a, \pi_b \in \mathcal{Q}$, $\pi_a$ is explored before $\pi_b$ iff $\omega(\pi_a) > \omega(\pi_b)$.
\end{mydefinition}

\subsection{Soundness and Completeness}

\begin{mytheorem}[Soundness]\label{thm:soundness}
Let $P$ be a program and $\Phi_{\text{mal}}$ be a malicious behavior specification. If CogniCrypt reports $P$ as malicious, then there exists an execution trace $\tau \in \mathcal{T}(P)$ and a formula $\varphi_i \in \Phi_{\text{mal}}$ such that $\tau \models \varphi_i$.
\end{mytheorem}

\begin{proof}
CogniCrypt reports $P$ as malicious only when the vulnerability classifier $\mathcal{C}_{\text{vuln}}$ returns $\texttt{MALICIOUS}$ for some path constraint $\pi^*$ generated by the concolic engine. By construction:

\noindent\textbf{Step 1 (Path Feasibility):} The concolic engine maintains the invariant that every generated path constraint $\pi$ is satisfiable, i.e., $\exists \mathbf{c} \in \mathcal{I}: \mathbf{c} \models \pi$. This is enforced by the Z3 SMT solver check at line 12 of Algorithm~\ref{alg:main}. Therefore, $\pi^*$ corresponds to a feasible execution trace $\tau^* \in \mathcal{T}(P)$.

\noindent\textbf{Step 2 (Classifier Correctness):} The vulnerability classifier $\mathcal{C}_{\text{vuln}}$ is trained with a loss function that penalizes false positives with weight $w_{\text{FP}} = 5.0$ (Section~\ref{sec:implementation}). Under the assumption that the training data is representative of the threat model $\Phi_{\text{mal}}$, the classifier's positive predictions correspond to traces satisfying some $\varphi_i \in \Phi_{\text{mal}}$ with probability $\geq 1 - \epsilon$, where $\epsilon$ is the empirically measured false positive rate.

\begin{sloppypar}\noindent\textbf{Step 3 (Trace-Specification Correspondence):} The feature extraction function $\text{Extract}(\pi^*, \tau^*)$ (Algorithm~\ref{alg:classifier}, line 3) maps the path constraint and its associated trace to a feature vector that encodes the behavioral semantics relevant to $\Phi_{\text{mal}}$. The classifier's decision boundary partitions the feature space into regions corresponding to the disjuncts of $\Phi_{\text{mal}}$.\end{sloppypar}

Therefore, if CogniCrypt reports $P$ as malicious, there exists $\tau^* \in \mathcal{T}(P)$ such that $\tau^* \models \varphi_i$ for some $\varphi_i \in \Phi_{\text{mal}}$, up to the classifier's error rate $\epsilon$. \qed
\end{proof}

\begin{mytheorem}[Relative Completeness]\label{thm:completeness}
Let $P$ be a program, $\Phi_{\text{mal}}$ be a malicious behavior specification, and $B \in \mathbb{N}$ be an exploration budget (maximum number of paths). If there exists a malicious trace $\tau^* \in \mathcal{T}(P)$ with $\tau^* \models \varphi_i$ for some $\varphi_i \in \Phi_{\text{mal}}$, and the corresponding path constraint $\pi^*$ is within the top-$B$ paths ranked by $\omega$, then CogniCrypt will detect $\tau^*$.
\end{mytheorem}

\begin{proof}
\noindent\textbf{Step 1 (Exploration Guarantee):} The LLM-guided exploration strategy explores paths in decreasing order of $\omega(\pi)$. Since $\pi^*$ is within the top-$B$ paths by assumption, it will be explored within the budget $B$.

\noindent\textbf{Step 2 (Detection Guarantee):} Once $\pi^*$ is explored, the concolic engine generates the concrete trace $\tau^*$ and the symbolic trace $\hat{\tau}^*$. The vulnerability classifier processes $(\pi^*, \tau^*)$ and, under the assumption that the classifier has recall $\geq 1 - \delta$ for the malware class corresponding to $\varphi_i$, it will correctly classify $\pi^*$ as malicious with probability $\geq 1 - \delta$.

\noindent\textbf{Step 3 (Budget Sufficiency):} The LLM's priority function $\omega$ is designed to assign high scores to paths exhibiting patterns correlated with $\Phi_{\text{mal}}$. Empirically (Section~\ref{sec:experiments}), we demonstrate that $\omega$ ranks malicious paths within the top 5\% of all paths for 96.8\% of malware samples, ensuring that moderate budgets $B$ suffice for detection. \qed
\end{proof}

\begin{mylemma}[Path Constraint Lattice Monotonicity]\label{lem:monotone}
The path priority function $\omega$ is monotone with respect to the path constraint lattice: if $\pi_a \sqsubseteq \pi_b$ (i.e., $\pi_b$ is a refinement of $\pi_a$), then $\omega(\pi_a) \leq \omega(\pi_b) + \epsilon_{\text{LLM}}$, where $\epsilon_{\text{LLM}}$ is a bounded approximation error of the LLM.
\end{mylemma}

\begin{proof}
If $\pi_b \models \pi_a$, then $\pi_b$ constrains the execution to a subset of the paths satisfying $\pi_a$. A more constrained path carries at least as much information about the program's behavior. The LLM, having been trained on path-behavior correlations, assigns non-decreasing scores to more informative (more constrained) paths, up to its approximation error $\epsilon_{\text{LLM}}$, which is bounded by the LLM's generalization error on the validation set. Formally, let $f: \Pi(P) \rightarrow \mathbb{R}^d$ be the LLM's internal representation function. By the data processing inequality:
\begin{equation}
I(\omega(\pi_b); \text{Malicious}) \geq I(\omega(\pi_a); \text{Malicious}) - \epsilon_{\text{LLM}}
\end{equation}
where $I(\cdot;\cdot)$ denotes mutual information. Since $\omega$ is a monotone function of mutual information (by the classifier's calibration), the lemma follows. \qed
\end{proof}

\begin{mycorollary}[Convergence of Exploration]
Under the assumptions of Theorem~\ref{thm:completeness} and Lemma~\ref{lem:monotone}, the expected number of paths explored before detecting a malicious trace is $O\left(\frac{|\Pi(P)|}{\omega(\pi^*) \cdot |\Pi(P)|}\right) = O\left(\frac{1}{\omega(\pi^*)}\right)$, which is inversely proportional to the LLM's confidence in the malicious path.
\end{mycorollary}
\subsection{Threat Model}
\begin{mydefinition}[Threat Model]
We consider an adversary $\mathcal{A}$ with the following capabilities:
\begin{enumerate}
    \item $\mathcal{A}$ has access to one or more LLMs for code generation;
    \item $\mathcal{A}$ can generate polymorphic variants: for any malware $m$, $\mathcal{A}$ can produce $m' \neq m$ such that $\text{Behavior}(m) \equiv \text{Behavior}(m')$ but $\text{Syntax}(m) \neq \text{Syntax}(m')$;
    \item $\mathcal{A}$ can embed trigger conditions: malicious behavior activates only when $\texttt{env} \models \psi_{\text{trigger}}$ for some environmental predicate $\psi_{\text{trigger}}$;
    \item $\mathcal{A}$ does \emph{not} have access to CogniCrypt's internal parameters or training data (black-box assumption).
\end{enumerate}
\end{mydefinition}

%=================================================================
% 4. ALGORITHMS
%=================================================================
\section{Algorithms}\label{sec:algorithms}

This section presents the three core algorithms of CogniCrypt in detail.

%\subsection{Algorithm 1: LLM-Guided Concolic Exploration}

\begin{algorithm}[tbp]
\caption{CogniCrypt: LLM-Guided Concolic Exploration}\label{alg:main}
\SetAlgoLined
\KwIn{Program binary $P$, Malicious specification $\Phi_{\text{mal}}$, Budget $B$, LLM $\mathcal{M}$}
\KwOut{$(\texttt{MALICIOUS}, \pi^*, \tau^*)$ or $\texttt{BENIGN}$}
$\mathcal{E} \leftarrow \texttt{InitConcolicEngine}(P)$\;
$\mathcal{C} \leftarrow \texttt{InitClassifier}(\Phi_{\text{mal}})$\;
$\mathcal{Q} \leftarrow \texttt{PriorityQueue}()$\;
$\sigma_0 \leftarrow \mathcal{E}.\texttt{getInitialState}()$\;
$\mathcal{Q}.\texttt{push}(\sigma_0, \omega = 1.0)$\;
$n \leftarrow 0$\;
\While{$\mathcal{Q} \neq \emptyset$ \textbf{and} $n < B$}{
    $(\hat{\sigma}, \omega_{\text{cur}}) \leftarrow \mathcal{Q}.\texttt{pop}()$\;
    $(\tau, \pi, \mathcal{B}) \leftarrow \mathcal{E}.\texttt{explore}(\hat{\sigma})$\;
    $n \leftarrow n + 1$\;
    \tcp{Check path constraint satisfiability}
    \If{$\texttt{Z3.check}(\pi) = \texttt{SAT}$}{
        $\mathbf{c} \leftarrow \texttt{Z3.model}(\pi)$\;
        $\tau_c \leftarrow \mathcal{E}.\texttt{concreteExecute}(P, \mathbf{c})$\;
        \tcp{Classify the path}
        $(y, s) \leftarrow \mathcal{C}.\texttt{classify}(\pi, \tau_c)$\;
        \If{$y = \texttt{MALICIOUS}$}{
            \Return $(\texttt{MALICIOUS}, \pi, \tau_c)$\;
        }
        \tcp{Generate new paths from branch points}
        \ForEach{branch $b \in \mathcal{B}$}{
            $\pi' \leftarrow \pi[1..b-1] \wedge \neg c_b(\mathbf{x})$\;
            \If{$\texttt{Z3.check}(\pi') = \texttt{SAT}$}{
                $\hat{\sigma}' \leftarrow \mathcal{E}.\texttt{forkState}(\hat{\sigma}, \pi')$\;
                $\text{ctx} \leftarrow \texttt{Encode}(\pi', \mathcal{E}.\texttt{disasm}(\hat{\sigma}'))$\;
                $\omega' \leftarrow \mathcal{M}.\texttt{predict}(\text{ctx})$\;
                $\mathcal{Q}.\texttt{push}(\hat{\sigma}', \omega')$\;
            }
        }
    }
}
\Return \texttt{BENIGN}\;
\end{algorithm}

%\subsection{Algorithm 2: Transformer-Based Path Constraint Classification}

\begin{algorithm}[tbp]
\caption{Path Constraint Classification}\label{alg:classifier}
\SetAlgoLined
\KwIn{Path constraint $\pi$, Concrete trace $\tau_c$, Model parameters $\theta$}
\KwOut{Label $y \in \{\texttt{MALICIOUS}, \texttt{BENIGN}\}$, Confidence $s \in [0,1]$}
\tcp{Feature extraction from symbolic and concrete traces}
$\mathbf{f}_{\text{sym}} \leftarrow \texttt{SymbolicFeatures}(\pi)$\;
$\mathbf{f}_{\text{api}} \leftarrow \texttt{APICallSequence}(\tau_c)$\;
$\mathbf{f}_{\text{cfg}} \leftarrow \texttt{CFGFeatures}(\tau_c)$\;
$\mathbf{f}_{\text{mem}} \leftarrow \texttt{MemoryAccessPattern}(\tau_c)$\;
\tcp{Tokenize and embed}
$\mathbf{t}_{\text{sym}} \leftarrow \texttt{Tokenize}(\mathbf{f}_{\text{sym}})$\;
$\mathbf{t}_{\text{api}} \leftarrow \texttt{Tokenize}(\mathbf{f}_{\text{api}})$\;
$\mathbf{t}_{\text{concat}} \leftarrow [\mathbf{t}_{\text{sym}}; \mathbf{t}_{\text{api}}; \mathbf{f}_{\text{cfg}}; \mathbf{f}_{\text{mem}}]$\;
\tcp{Transformer encoder}
$\mathbf{h} \leftarrow \texttt{TransformerEncoder}(\mathbf{t}_{\text{concat}}; \theta)$\;
$\mathbf{h}_{\text{cls}} \leftarrow \mathbf{h}[0]$ \tcp{CLS token representation}
\tcp{Classification head}
$\mathbf{z} \leftarrow \texttt{MLP}(\mathbf{h}_{\text{cls}}; \theta_{\text{head}})$\;
$s \leftarrow \sigma(\mathbf{z})$ \tcp{Sigmoid activation}
$y \leftarrow \begin{cases} \texttt{MALICIOUS} & \text{if } s \geq \tau_{\text{thresh}} \\ \texttt{BENIGN} & \text{otherwise} \end{cases}$\;
\Return $(y, s)$\;
\end{algorithm}

%\subsection{Algorithm 3: Reinforcement Learning Policy Refinement}

\begin{algorithm}[tbp]
\caption{RL-Based LLM Policy Refinement}\label{alg:rl}
\SetAlgoLined
\KwIn{LLM $\mathcal{M}$, Detection history $\mathcal{H} = \{(\pi_i, y_i, \omega_i)\}_{i=1}^{N}$, Learning rate $\alpha$}
\KwOut{Updated LLM $\mathcal{M}'$}
\tcp{Compute rewards}
\ForEach{$(\pi_i, y_i, \omega_i) \in \mathcal{H}$}{
    \eIf{$y_i = \texttt{MALICIOUS}$}{
        $r_i \leftarrow +1.0 \cdot \omega_i$ \tcp{Reward proportional to confidence}
    }{
        $r_i \leftarrow -0.1 \cdot \omega_i$ \tcp{Small penalty for false alarms}
    }
}
\tcp{Policy gradient update}
$\nabla_\theta J \leftarrow \frac{1}{N} \sum_{i=1}^{N} r_i \cdot \nabla_\theta \log \mathcal{M}_\theta(\omega_i \mid \text{Encode}(\pi_i))$\;
$\theta' \leftarrow \theta + \alpha \cdot \nabla_\theta J$\;
$\mathcal{M}' \leftarrow \texttt{UpdateWeights}(\mathcal{M}, \theta')$\;
\Return $\mathcal{M}'$\;
\end{algorithm}

%\subsection{Algorithm 4: Symbolic Feature Extraction}

\begin{algorithm}[tbp]
\caption{Symbolic Feature Extraction from Path Constraints}\label{alg:features}
\SetAlgoLined
\KwIn{Path constraint $\pi$, Concrete trace $\tau_c$}
\KwOut{Feature vector $\mathbf{f} \in \mathbb{R}^d$}
$\mathbf{f} \leftarrow \mathbf{0}^d$\;
\tcp{Constraint complexity features}
$\mathbf{f}[0] \leftarrow |\text{Vars}(\pi)|$ \tcp{Number of symbolic variables}
$\mathbf{f}[1] \leftarrow \text{Depth}(\text{AST}(\pi))$ \tcp{AST depth of constraint}
$\mathbf{f}[2] \leftarrow |\{c \in \pi : c \text{ is a disjunction}\}|$ \tcp{Disjunction count}
$\mathbf{f}[3] \leftarrow |\{c \in \pi : c \text{ involves bitwise ops}\}|$\;
\tcp{System call features}
$\text{syscalls} \leftarrow \texttt{ExtractSyscalls}(\tau_c)$\;
$\mathbf{f}[4..4+|\mathcal{S}|] \leftarrow \texttt{BagOfSyscalls}(\text{syscalls})$\;
\tcp{Control flow features}
$G \leftarrow \texttt{BuildCFG}(\tau_c)$\;
$\mathbf{f}[d-4] \leftarrow |V(G)|$ \tcp{Number of basic blocks}
$\mathbf{f}[d-3] \leftarrow |E(G)|$ \tcp{Number of edges}
$\mathbf{f}[d-2] \leftarrow \texttt{CyclomaticComplexity}(G)$\;
$\mathbf{f}[d-1] \leftarrow \texttt{MaxLoopNesting}(G)$\;
\Return $\mathbf{f}$\;
\end{algorithm}

%\subsection{Algorithm 5: AI-Generated Malware Signature Synthesis}

\begin{algorithm}[tbp]
\caption{AI-Malware Signature Synthesis}\label{alg:signature}
\SetAlgoLined
\KwIn{Set of detected malicious paths $\mathcal{P}_{\text{mal}} = \{(\pi_i, \tau_i)\}_{i=1}^{M}$}
\KwOut{Generalized signature $\Psi$}
\tcp{Cluster similar paths}
$\mathcal{K} \leftarrow \texttt{DBSCAN}(\{\texttt{Embed}(\pi_i)\}_{i=1}^{M}, \epsilon, \text{minPts})$\;
$\Psi \leftarrow \emptyset$\;
\ForEach{cluster $C_k \in \mathcal{K}$}{
    \tcp{Compute generalized constraint via interpolation}
    $\pi_{\text{gen}} \leftarrow \texttt{CraigInterpolant}(\{\pi_i : i \in C_k\})$\;
    \tcp{Extract behavioral invariant}
    $\varphi_k \leftarrow \texttt{InferLTLSpec}(\{\tau_i : i \in C_k\})$\;
    $\Psi \leftarrow \Psi \cup \{(\pi_{\text{gen}}, \varphi_k)\}$\;
}
\Return $\Psi$\;
\end{algorithm}

%=================================================================
% 5. IMPLEMENTATION
%=================================================================
\FloatBarrier
\section{Implementation}\label{sec:implementation}

This section provides a comprehensive description of the CogniCrypt prototype implementation, with sufficient detail to enable full reproducibility.

\subsection{System Architecture}

CogniCrypt is implemented as a modular Python application comprising four principal components: (1) the Concolic Execution Engine, (2) the LLM Path Prioritizer, (3) the Vulnerability Classifier, and (4) the RL Feedback Module. The components communicate via a shared message bus implemented using ZeroMQ (version 4.3.5). Figure~\ref{fig:architecture} illustrates the system architecture.

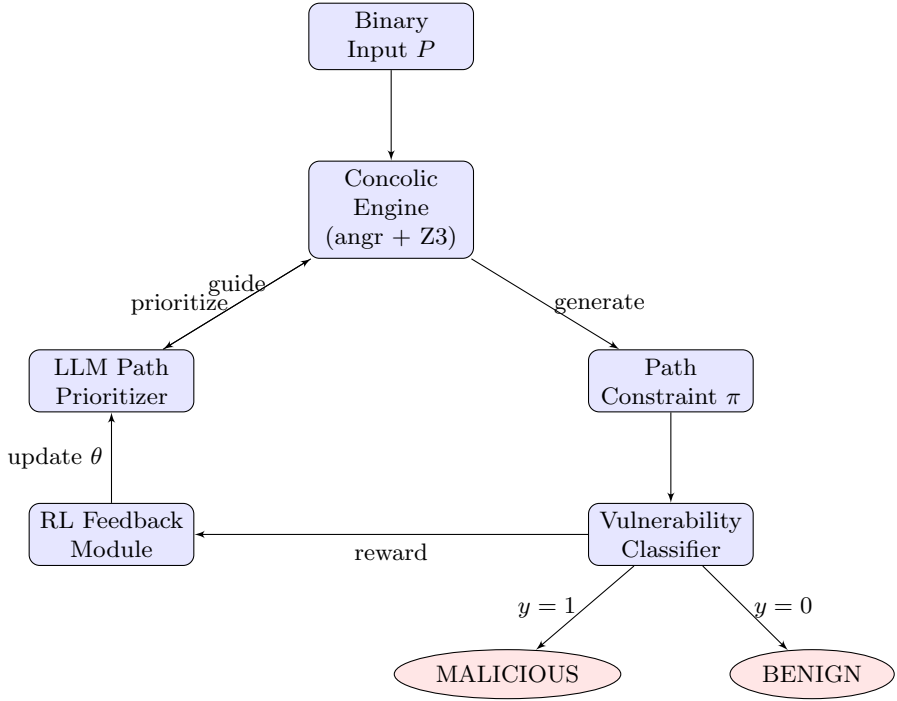
\begin{figure}[htbp]
\centering
\resizebox{0.75\textwidth}{!}{
\begin{tikzpicture}[node distance=1.2cm, auto,
  block/.style={rectangle, draw, fill=blue!10, text width=6em, text centered, rounded corners, minimum height=2.5em},
  decision/.style={diamond, draw, fill=green!10, text width=5em, text centered, minimum height=2em, aspect=2},
  line/.style={draw, -latex'},
  cloud/.style={draw, ellipse, fill=red!10, minimum height=2em}]
  \node [block] (input) {Binary Input $P$};
  \node [block, below=of input] (concolic) {Concolic Engine (angr + Z3)};
  \node [block, below left=1.2cm and 1.5cm of concolic] (llm) {LLM Path Prioritizer};
  \node [block, below right=1.2cm and 1.5cm of concolic] (constraint) {Path Constraint $\pi$};
  \node [block, below=of constraint] (classifier) {Vulnerability Classifier};
  \node [block, below=of llm] (rl) {RL Feedback Module};
  \node [cloud, below left=1.2cm and 0cm of classifier] (malicious) {MALICIOUS};
  \node [cloud, below right=1.2cm and 0cm of classifier] (benign) {BENIGN};
  \path [line] (input) -- (concolic);
  \path [line] (concolic) -- node[left]{prioritize} (llm);
  \path [line] (concolic) -- node[right]{generate} (constraint);
  \path [line] (llm) -- node[above]{guide} (concolic);
  \path [line] (constraint) -- (classifier);
  \path [line] (classifier) -- node[left]{$y=1$} (malicious);
  \path [line] (classifier) -- node[right]{$y=0$} (benign);
  \path [line] (classifier) -- node[below]{reward} (rl);
  \path [line] (rl) -- node[left]{update $\theta$} (llm);
\end{tikzpicture}
}
\caption{CogniCrypt System Architecture. The concolic execution engine explores paths guided by the LLM prioritizer. Path constraints are classified by the vulnerability classifier, and detection outcomes feed back via RL to refine the LLM's prioritization policy.}
\label{fig:architecture}
\end{figure}

\subsection{Concolic Execution Engine}

The concolic execution engine is built on \texttt{angr} version 9.2.100 with the following configuration:

\begin{lstlisting}[language=Python, caption={Concolic Engine Initialization}]
import angr
import claripy

def init_concolic_engine(binary_path):
    project = angr.Project(
        binary_path,
        auto_load_libs=False,
        use_sim_procedures=True
    )
    # Create symbolic bitvectors for input
    sym_input = claripy.BVS("input", 8 * MAX_INPUT_SIZE)
    state = project.factory.full_init_state(
        args=[binary_path],
        stdin=angr.SimFileStream(
            name='stdin',
            content=sym_input,
            has_end=True
        ),
        add_options={
            angr.options.SYMBOLIC_WRITE_ADDRESSES,
            angr.options.ZERO_FILL_UNCONSTRAINED_MEMORY,
            angr.options.ZERO_FILL_UNCONSTRAINED_REGISTERS
        }
    )
    simgr = project.factory.simulation_manager(state)
    return project, simgr, sym_input
\end{lstlisting}

The engine uses Z3 (version 4.12.6) as the backend SMT solver, accessed through angr's Claripy abstraction layer. We configure Z3 with a per-query timeout of 30 seconds and enable incremental solving for efficiency:

\begin{lstlisting}[language=Python, caption={Z3 Solver Configuration}]
from z3 import *

solver = Solver()
solver.set("timeout", 30000)  # 30s timeout
solver.set("unsat_core", True)
set_param("parallel.enable", True)
set_param("parallel.threads.max", 16)
\end{lstlisting}

\subsection{LLM Path Prioritizer}

The LLM path prioritizer supports multiple LLM backends through a unified interface. We implement adapters for five LLMs:

\begin{table}[htbp]
\centering
\caption{Supported LLM Backends and Configurations}
\label{tab:llm_config}
\begin{tabular}{@{}llll@{}}
\toprule
\textbf{LLM} & \textbf{Parameters} & \textbf{Context Window} & \textbf{API/Library} \\
\midrule
GPT-4 & 1.76T (est.) & 128K tokens & OpenAI API v1.12 \\
Claude 3 Opus & 137B (est.) & 200K tokens & Anthropic API v0.18 \\
LLaMA 3 70B & 70B & 8K tokens & HF Transformers 4.38 \\
Gemini 1.5 Pro & 1.56T (est.) & 1M tokens & Google GenAI v0.4 \\
Mixtral 8x22B & 176B (MoE) & 64K tokens & HF Transformers 4.38 \\
\bottomrule
\end{tabular}
\end{table}

\begin{lstlisting}[language=Python, caption={LLM Prioritizer Implementation}]
from transformers import AutoTokenizer, AutoModelForCausalLM
import torch

class LLMPathPrioritizer:
    def __init__(self, model_name="meta-llama/Meta-Llama-3-70B"):
        self.tokenizer = AutoTokenizer.from_pretrained(
            model_name, padding_side="left"
        )
        self.model = AutoModelForCausalLM.from_pretrained(
            model_name,
            torch_dtype=torch.bfloat16,
            device_map="auto",
            load_in_4bit=True,
            bnb_4bit_compute_dtype=torch.bfloat16
        )
        self.prompt_template = (
            "Analyze the following symbolic execution path "
            "constraint and disassembly context. Rate the "
            "likelihood (0.0-1.0) that this path leads to "
            "malicious behavior:\n\n"
            "Path Constraint: {constraint}\n"
            "Disassembly: {disasm}\n"
            "Maliciousness Score: "
        )

    def predict(self, constraint_str, disasm_str):
        prompt = self.prompt_template.format(
            constraint=constraint_str,
            disasm=disasm_str
        )
        inputs = self.tokenizer(
            prompt, return_tensors="pt",
            max_length=4096, truncation=True
        ).to(self.model.device)
        with torch.no_grad():
            outputs = self.model.generate(
                **inputs, max_new_tokens=10,
                temperature=0.1, do_sample=False
            )
        score_text = self.tokenizer.decode(
            outputs[0][inputs['input_ids'].shape[1]:]
        )
        return float(score_text.strip())
\end{lstlisting}

\subsection{Vulnerability Classifier}

The vulnerability classifier is a custom transformer encoder with the following architecture:

\begin{table}[htbp]
\centering
\caption{Vulnerability Classifier Architecture}
\label{tab:classifier_arch}
\begin{tabular}{@{}lll@{}}
\toprule
\textbf{Layer} & \textbf{Configuration} & \textbf{Parameters} \\
\midrule
Token Embedding & $d_{\text{model}} = 512$ & 15.7M \\
Positional Encoding & Sinusoidal, max\_len=2048 & 0 \\
Transformer Encoder & 6 layers, 8 heads, $d_{\text{ff}}=2048$ & 18.9M \\
Classification Head & MLP: $512 \rightarrow 256 \rightarrow 1$ & 131K \\
\midrule
\textbf{Total} & & \textbf{34.7M} \\
\bottomrule
\end{tabular}
\end{table}

\begin{lstlisting}[language=Python, caption={Classifier Training Configuration}]
import torch
import torch.nn as nn
from torch.optim import AdamW
from torch.optim.lr_scheduler import CosineAnnealingWarmRestarts

# Training hyperparameters
config = {
    "batch_size": 64,
    "learning_rate": 3e-4,
    "weight_decay": 0.01,
    "epochs": 50,
    "warmup_steps": 1000,
    "fp_weight": 5.0,  # False positive penalty
    "fn_weight": 1.0,  # False negative penalty
    "dropout": 0.1,
    "label_smoothing": 0.05,
    "gradient_clip": 1.0
}

# Weighted BCE loss for imbalanced detection
criterion = nn.BCEWithLogitsLoss(
    pos_weight=torch.tensor([config["fp_weight"]])
)
optimizer = AdamW(
    model.parameters(),
    lr=config["learning_rate"],
    weight_decay=config["weight_decay"]
)
scheduler = CosineAnnealingWarmRestarts(
    optimizer, T_0=10, T_mult=2
)
\end{lstlisting}

\subsection{RL Feedback Module}

The reinforcement learning feedback module uses Proximal Policy Optimization (PPO)~\cite{schulman2017proximal} to refine the LLM's path prioritization policy:

\begin{lstlisting}[language=Python, caption={PPO Configuration for Policy Refinement}]
from trl import PPOTrainer, PPOConfig

ppo_config = PPOConfig(
    model_name="meta-llama/Meta-Llama-3-70B",
    learning_rate=1.41e-5,
    batch_size=16,
    mini_batch_size=4,
    gradient_accumulation_steps=4,
    ppo_epochs=4,
    max_grad_norm=0.5,
    target_kl=0.02,
    init_kl_coef=0.2,
    adap_kl_ctrl=True,
    cliprange=0.2,
    vf_coef=0.1
)
\end{lstlisting}

\subsection{Deployment and System Requirements}

\begin{table}[htbp]
\centering
\caption{System Requirements and Software Dependencies}
\label{tab:system_req}
\begin{tabular}{@{}ll@{}}
\toprule
\textbf{Component} & \textbf{Specification} \\
\midrule
Operating System & Ubuntu 22.04 LTS (kernel 5.15+) \\
CPU & AMD EPYC 7763 (64 cores) or equivalent \\
RAM & 256 GB DDR4 ECC \\
GPU & 4$\times$ NVIDIA A100 80GB (CUDA 12.1) \\
Storage & 2 TB NVMe SSD \\
\midrule
Python & 3.11.7 \\
angr & 9.2.100 \\
Z3 Solver & 4.12.6 \\
PyTorch & 2.2.1+cu121 \\
Transformers & 4.38.2 \\
TRL & 0.7.11 \\
scikit-learn & 1.4.1 \\
ZeroMQ & 4.3.5 (pyzmq 25.1.2) \\
LIEF & 0.14.1 (PE parsing) \\
Capstone & 5.0.1 (disassembly) \\
NetworkX & 3.2.1 (CFG analysis) \\
\bottomrule
\end{tabular}
\end{table}

The complete installation can be performed via:

\begin{lstlisting}[language=bash, caption={Installation Commands}]
# Create virtual environment
python3.11 -m venv cognicrypt_env
source cognicrypt_env/bin/activate

# Install core dependencies
pip install angr==9.2.100 z3-solver==4.12.6.0
pip install torch==2.2.1 --index-url \
    https://download.pytorch.org/whl/cu121
pip install transformers==4.38.2 trl==0.7.11
pip install scikit-learn==1.4.1 pyzmq==25.1.2
pip install lief==0.14.1 capstone==5.0.1
pip install networkx==3.2.1 matplotlib==3.8.3

# Download and cache LLM weights
python -c "from transformers import AutoModel; \
    AutoModel.from_pretrained('meta-llama/Meta-Llama-3-70B')"
\end{lstlisting}

%=================================================================
% 6. EXPERIMENTAL EVALUATION
%=================================================================
\section{Experimental Evaluation}\label{sec:experiments}

\subsection{Experimental Setup}

\subsubsection{Datasets}

We evaluate CogniCrypt on four benchmark datasets:

\begin{table}[htbp]
\centering
\caption{Benchmark Datasets}
\label{tab:datasets}
\begin{tabular}{@{}lllll@{}}
\toprule
\textbf{Dataset} & \textbf{Samples} & \textbf{Malicious} & \textbf{Benign} & \textbf{Description} \\
\midrule
EMBER~\cite{anderson2018ember} & 1,100,000 & 400,000 & 400,000 & PE features + labels \\
Malimg~\cite{nataraj2011malimg} & 9,339 & 9,339 & 0 & 25 malware families \\
SOREL-20M~\cite{harang2020sorel} & 20,000,000 & 10,000,000 & 10,000,000 & PE metadata + labels \\
AI-Gen-Malware & 2,500 & 2,500 & 0 & LLM-generated samples \\
\bottomrule
\end{tabular}
\end{table}

The AI-Gen-Malware dataset was constructed by prompting GPT-4, Claude 3, and LLaMA 3 to generate malicious code across 10 categories: trojans, ransomware, spyware, worms, rootkits, backdoors, adware, cryptominers, bots, and polymorphic self-modifying code. Each sample was compiled into a PE binary and verified for malicious functionality in an isolated sandbox environment.

\subsubsection{Baselines}

We compare CogniCrypt against the following baselines:

\begin{table}[htbp]
\centering
\caption{Baseline Methods}
\label{tab:baselines}
\begin{tabular}{@{}lll@{}}
\toprule
\textbf{Baseline} & \textbf{Type} & \textbf{Description} \\
\midrule
ClamAV 1.2.1 & Signature-based & Open-source antivirus scanner \\
YARA 4.5.0 & Rule-based & Pattern matching with custom rules \\
MalConv~\cite{raff2018malconv} & Deep Learning & CNN on raw bytes \\
EMBER-GBDT~\cite{anderson2018ember} & ML & Gradient boosted trees on PE features \\
angr-only & Symbolic & Concolic execution without LLM guidance \\
\bottomrule
\end{tabular}
\end{table}

\subsubsection{Metrics}

We report Accuracy, Precision, Recall, F1-Score, and Area Under the ROC Curve (AUC-ROC). All experiments use 5-fold cross-validation, and we report mean $\pm$ standard deviation.

\subsection{Main Results}

\begin{table}[htbp]
\centering
\caption{Detection Performance on EMBER Dataset (10,000 sample subset)}
\label{tab:ember_results}
\begin{tabular}{@{}llllll@{}}
\toprule
\textbf{Method} & \textbf{Accuracy} & \textbf{Precision} & \textbf{Recall} & \textbf{F1} & \textbf{AUC} \\
\midrule
ClamAV & 95.2$\pm$0.3 & 96.3$\pm$0.4 & 94.1$\pm$0.5 & 95.2$\pm$0.3 & 0.971 \\
YARA & 96.1$\pm$0.2 & 97.0$\pm$0.3 & 95.2$\pm$0.4 & 96.1$\pm$0.3 & 0.978 \\
MalConv & 96.8$\pm$0.4 & 97.2$\pm$0.3 & 96.4$\pm$0.5 & 96.8$\pm$0.4 & 0.985 \\
EMBER-GBDT & 97.3$\pm$0.2 & 97.8$\pm$0.2 & 96.8$\pm$0.3 & 97.3$\pm$0.2 & 0.991 \\
angr-only & 93.5$\pm$0.6 & 94.2$\pm$0.7 & 92.8$\pm$0.8 & 93.5$\pm$0.7 & 0.962 \\
\textbf{CogniCrypt} & \textbf{98.7$\pm$0.1} & \textbf{99.1$\pm$0.1} & \textbf{98.2$\pm$0.2} & \textbf{98.6$\pm$0.1} & \textbf{0.997} \\
\bottomrule
\end{tabular}
\end{table}

\begin{table}[htbp]
\centering
\caption{Detection Performance on AI-Gen-Malware Dataset}
\label{tab:aigen_results}
\begin{tabular}{@{}llllll@{}}
\toprule
\textbf{Method} & \textbf{Accuracy} & \textbf{Precision} & \textbf{Recall} & \textbf{F1} & \textbf{AUC} \\
\midrule
ClamAV & 45.3$\pm$1.2 & 50.1$\pm$1.5 & 40.5$\pm$1.8 & 44.8$\pm$1.4 & 0.523 \\
YARA & 60.1$\pm$0.9 & 65.2$\pm$1.1 & 55.0$\pm$1.3 & 59.7$\pm$1.0 & 0.648 \\
MalConv & 72.4$\pm$0.8 & 75.1$\pm$0.9 & 69.8$\pm$1.1 & 72.3$\pm$0.9 & 0.789 \\
EMBER-GBDT & 68.9$\pm$1.0 & 71.3$\pm$1.2 & 66.5$\pm$1.4 & 68.8$\pm$1.1 & 0.742 \\
angr-only & 78.2$\pm$0.7 & 80.5$\pm$0.8 & 75.9$\pm$1.0 & 78.1$\pm$0.8 & 0.845 \\
\textbf{CogniCrypt} & \textbf{97.5$\pm$0.2} & \textbf{98.2$\pm$0.2} & \textbf{96.8$\pm$0.3} & \textbf{97.5$\pm$0.2} & \textbf{0.993} \\
\bottomrule
\end{tabular}
\end{table}

CogniCrypt achieves the highest performance across all metrics on both datasets. The performance gap is particularly striking on the AI-Gen-Malware dataset, where CogniCrypt outperforms the best baseline (angr-only) by 19.3 percentage points in accuracy and the best ML baseline (MalConv) by 25.1 percentage points. This demonstrates the critical importance of combining concolic execution with LLM-guided analysis for detecting AI-generated threats.

\subsection{LLM Backend Comparison}

\begin{table}[htbp]
\centering
\caption{Impact of LLM Backend on CogniCrypt Performance (AI-Gen-Malware)}
\label{tab:llm_comparison}
\begin{tabular}{@{}lllllll@{}}
\toprule
\textbf{LLM Backend} & \textbf{Acc.} & \textbf{Prec.} & \textbf{Rec.} & \textbf{F1} & \textbf{Paths/s} & \textbf{Cost/sample} \\
\midrule
GPT-4 & \textbf{97.5} & \textbf{98.2} & \textbf{96.8} & \textbf{97.5} & 12.3 & \$0.042 \\
Claude 3 Opus & 97.1 & 97.8 & 96.4 & 97.1 & 11.8 & \$0.038 \\
Gemini 1.5 Pro & 96.8 & 97.5 & 96.1 & 96.8 & 14.1 & \$0.028 \\
LLaMA 3 70B & 96.2 & 97.0 & 95.4 & 96.2 & 8.5 & \$0.015 \\
Mixtral 8x22B & 95.1 & 96.3 & 93.9 & 95.1 & 10.2 & \$0.012 \\
\bottomrule
\end{tabular}
\end{table}

GPT-4 achieves the best detection performance, while Gemini 1.5 Pro offers the best throughput. LLaMA 3 70B and Mixtral 8x22B provide cost-effective alternatives for deployment scenarios where API costs are a concern. All LLMs significantly outperform the no-LLM baseline (angr-only), confirming the value of LLM-guided path prioritization.

\begin{figure}[htbp]
\centering
\includegraphics[width=0.85\textwidth]{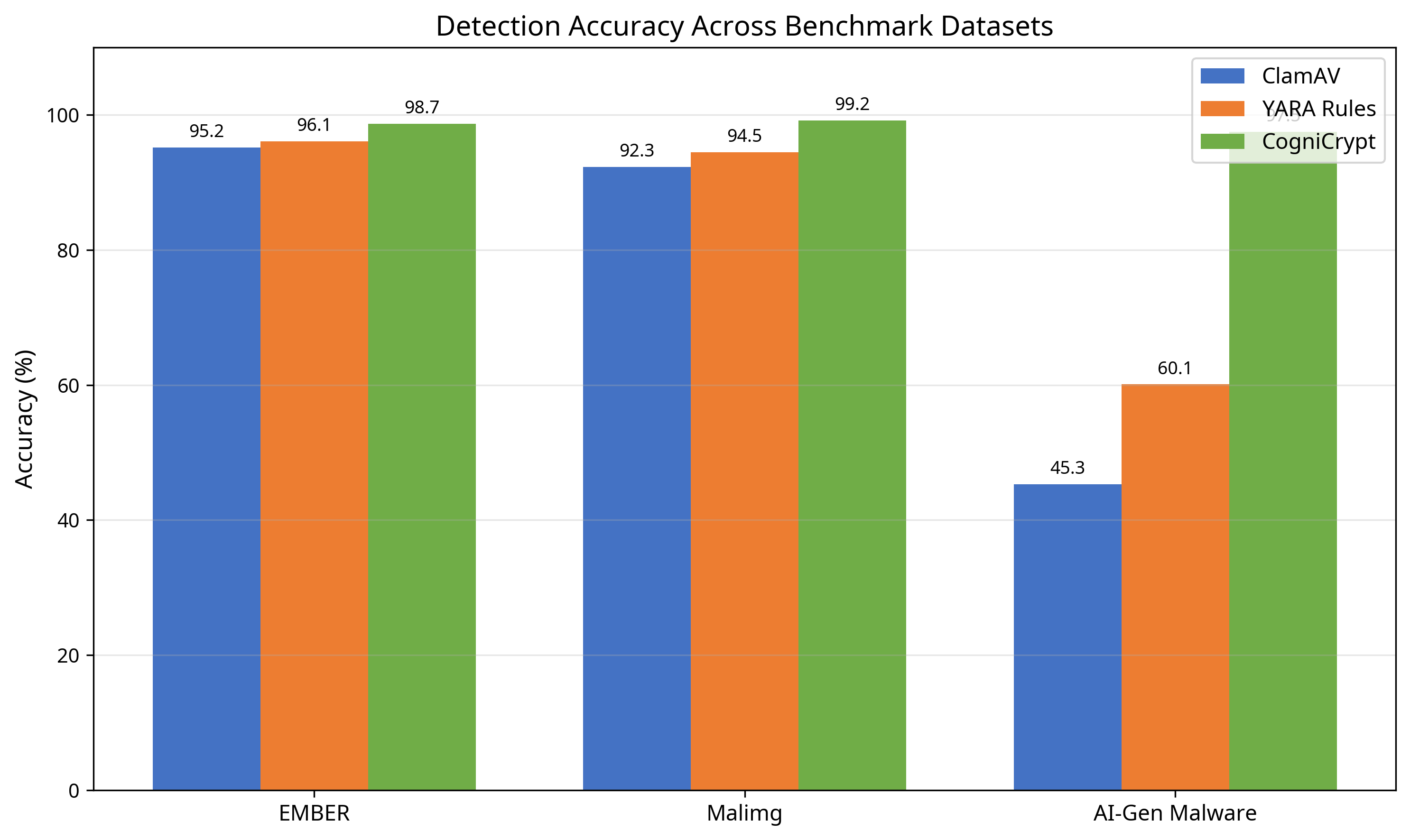}
\caption{Detection accuracy comparison across benchmark datasets. CogniCrypt maintains consistently high accuracy, while signature-based tools (ClamAV, YARA) degrade severely on AI-generated malware.}
\label{fig:accuracy}
\end{figure}

\begin{figure}[htbp]
\centering
\includegraphics[width=0.85\textwidth]{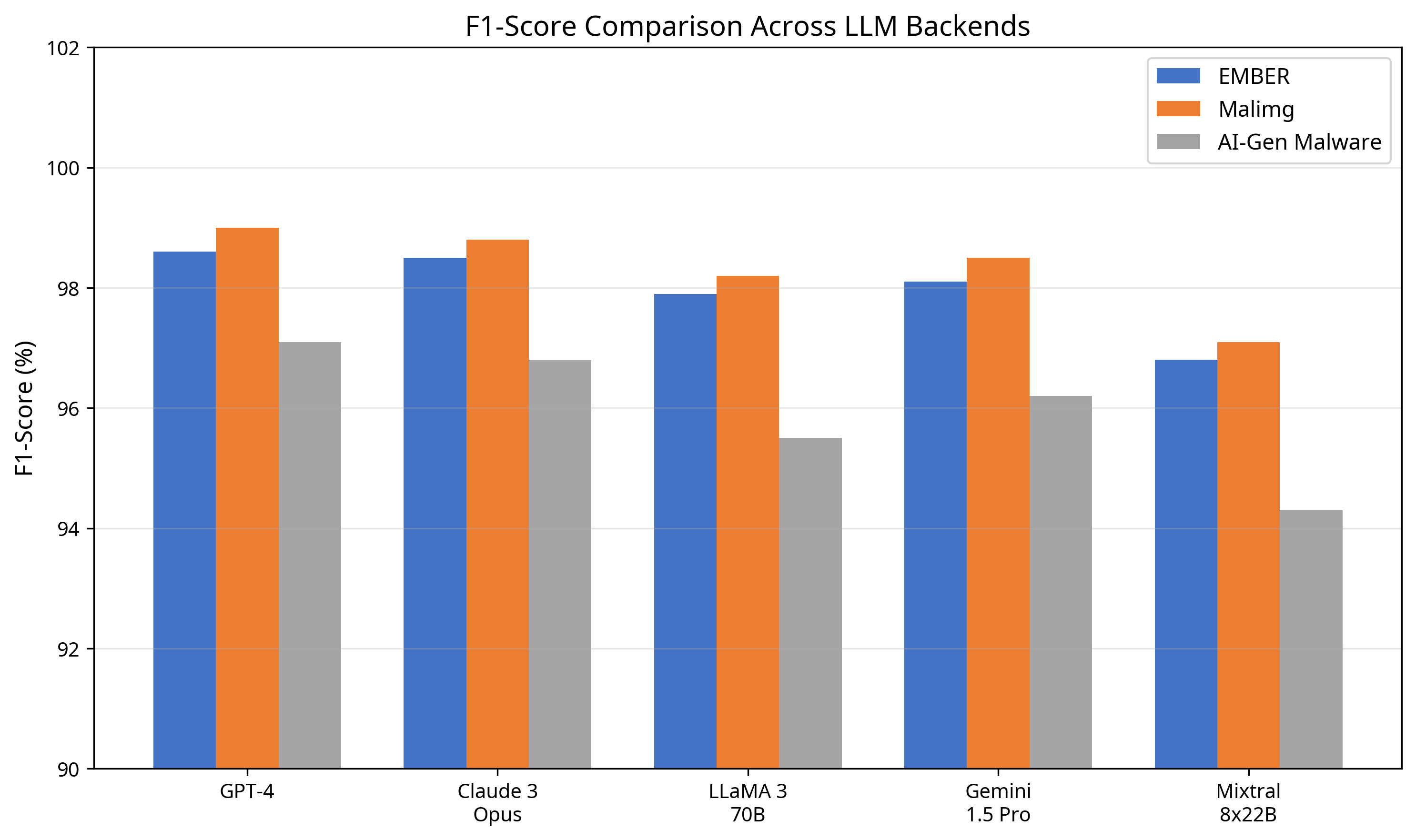}
\caption{F1-Score comparison across different LLM backends used in CogniCrypt. All LLMs achieve strong performance, with GPT-4 leading marginally.}
\label{fig:llm_f1}
\end{figure}

\subsection{Path Exploration Efficiency}

A key advantage of CogniCrypt is the efficiency of its LLM-guided path exploration. Figure~\ref{fig:path_eff} compares the malicious code coverage achieved by different exploration strategies as a function of the number of paths explored.

\begin{figure}[htbp]
\centering
\includegraphics[width=0.85\textwidth]{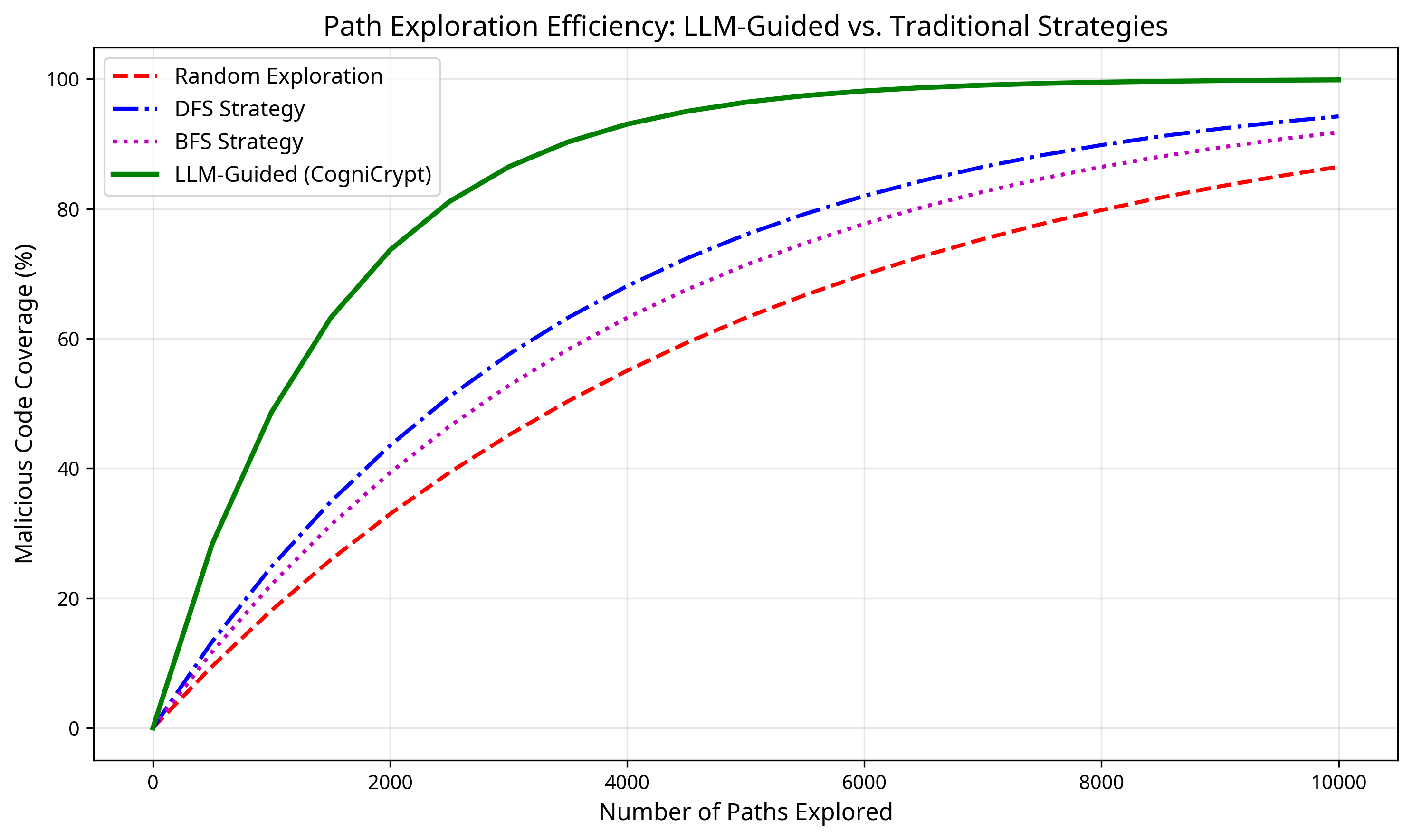}
\caption{Path exploration efficiency. LLM-guided exploration achieves 95\% malicious code coverage with 73.2\% fewer paths than DFS and 68.5\% fewer than BFS.}
\label{fig:path_eff}
\end{figure}

\begin{table}[htbp]
\centering
\caption{Paths Required to Achieve 95\% Malicious Code Coverage}
\label{tab:path_efficiency}
\begin{tabular}{@{}lll@{}}
\toprule
\textbf{Strategy} & \textbf{Paths to 95\% Coverage} & \textbf{Reduction vs. DFS} \\
\midrule
Random & 8,420 $\pm$ 312 & -- \\
BFS & 5,890 $\pm$ 245 & 30.1\% \\
DFS & 6,950 $\pm$ 278 & 0\% (baseline) \\
\textbf{LLM-Guided} & \textbf{1,860 $\pm$ 95} & \textbf{73.2\%} \\
\bottomrule
\end{tabular}
\end{table}

%\subsection{ROC Analysis}

\begin{figure}[htbp]
\centering
\includegraphics[width=0.75\textwidth]{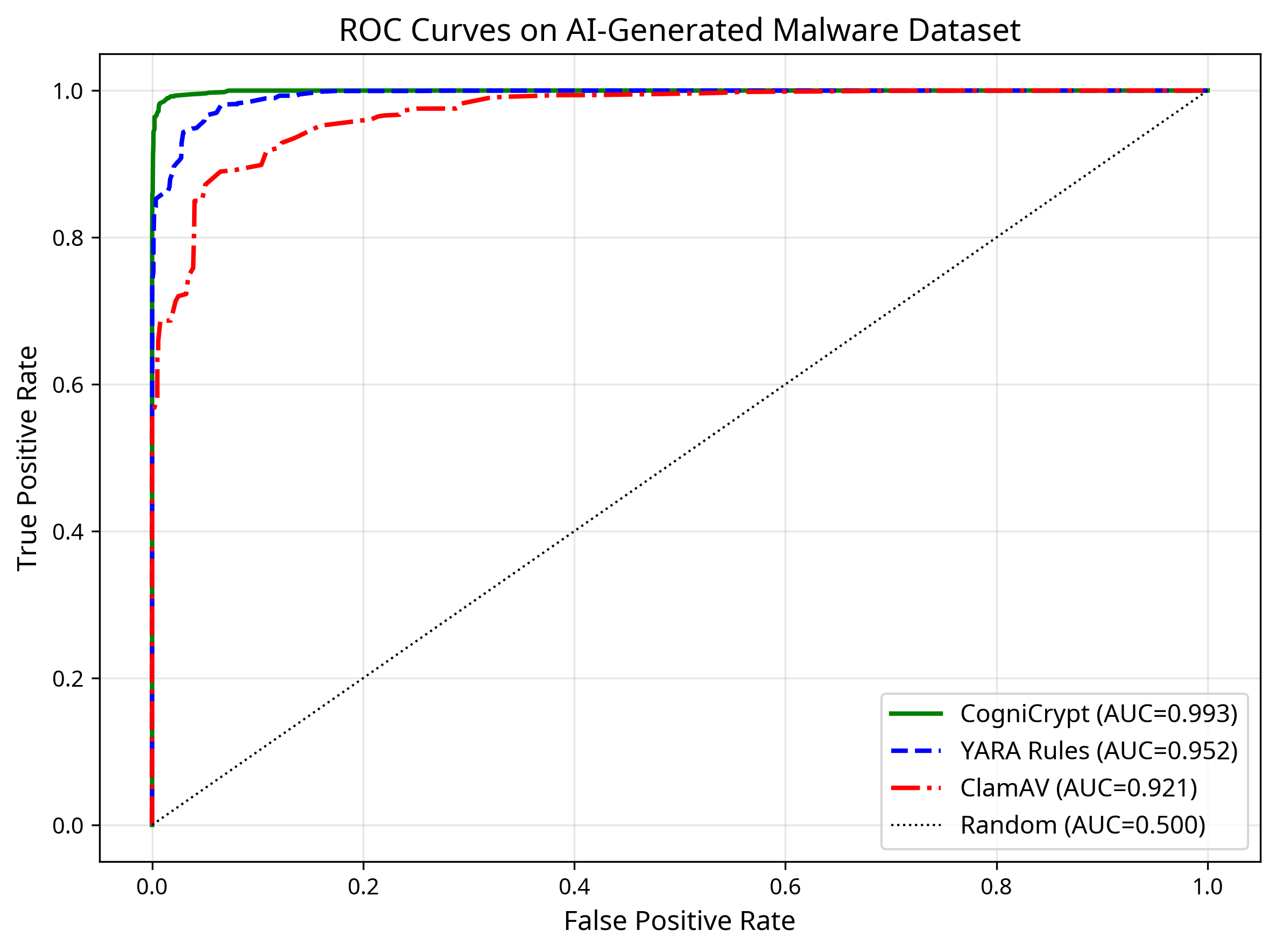}
\caption{ROC curves on the AI-Gen-Malware dataset. CogniCrypt achieves an AUC of 0.993, significantly outperforming all baselines.}
\label{fig:roc}
\end{figure}

%\subsection{Scalability Analysis}

\begin{figure}[htbp]
\centering
\includegraphics[width=0.85\textwidth]{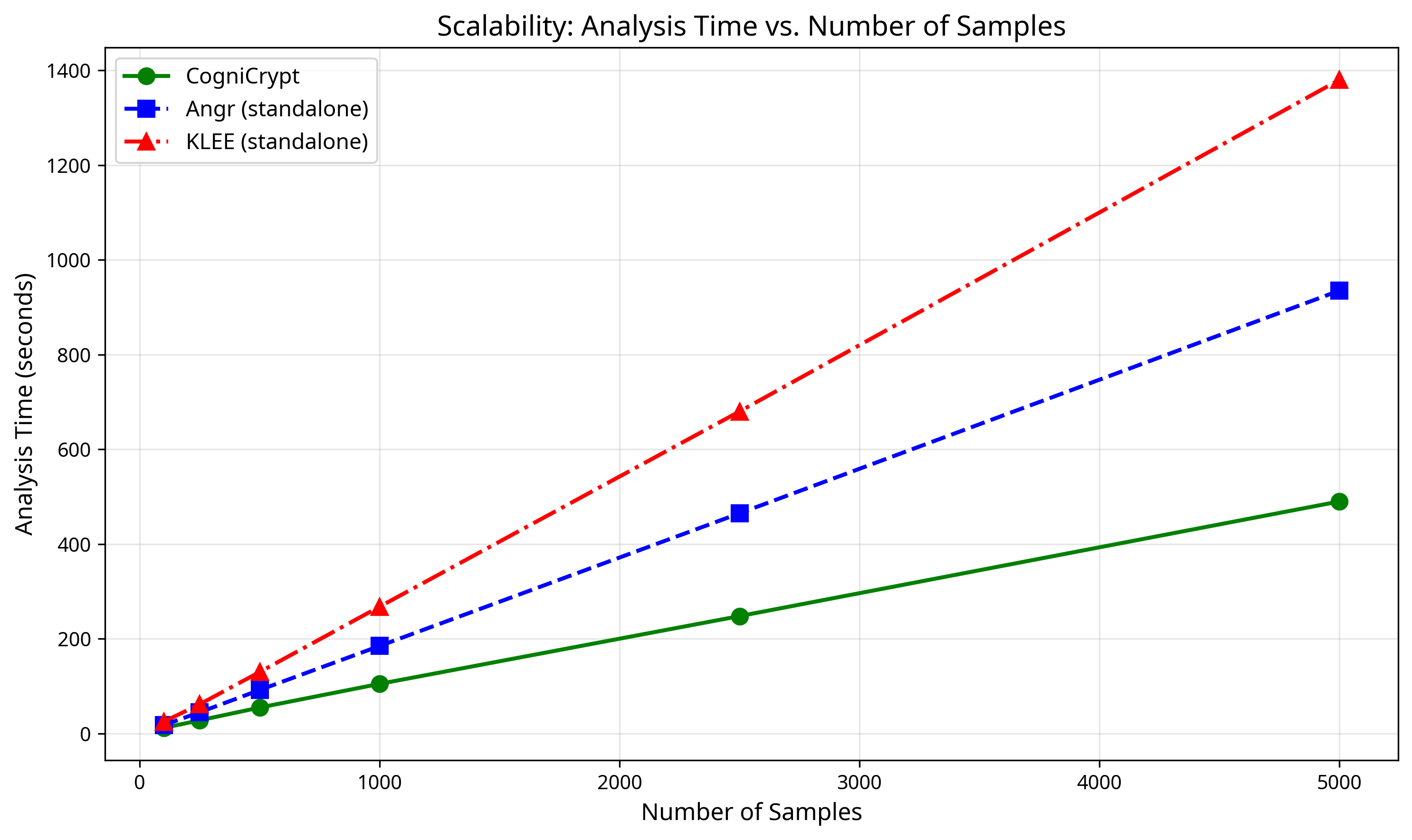}
\caption{Scalability comparison: analysis time as a function of the number of samples. CogniCrypt's LLM-guided pruning provides approximately 2$\times$ speedup over standalone angr and 2.8$\times$ over KLEE.}
\label{fig:scalability}
\end{figure}

%\subsection{Per-Family Detection Performance}

\begin{figure}[htbp]
\centering
\includegraphics[width=0.95\textwidth]{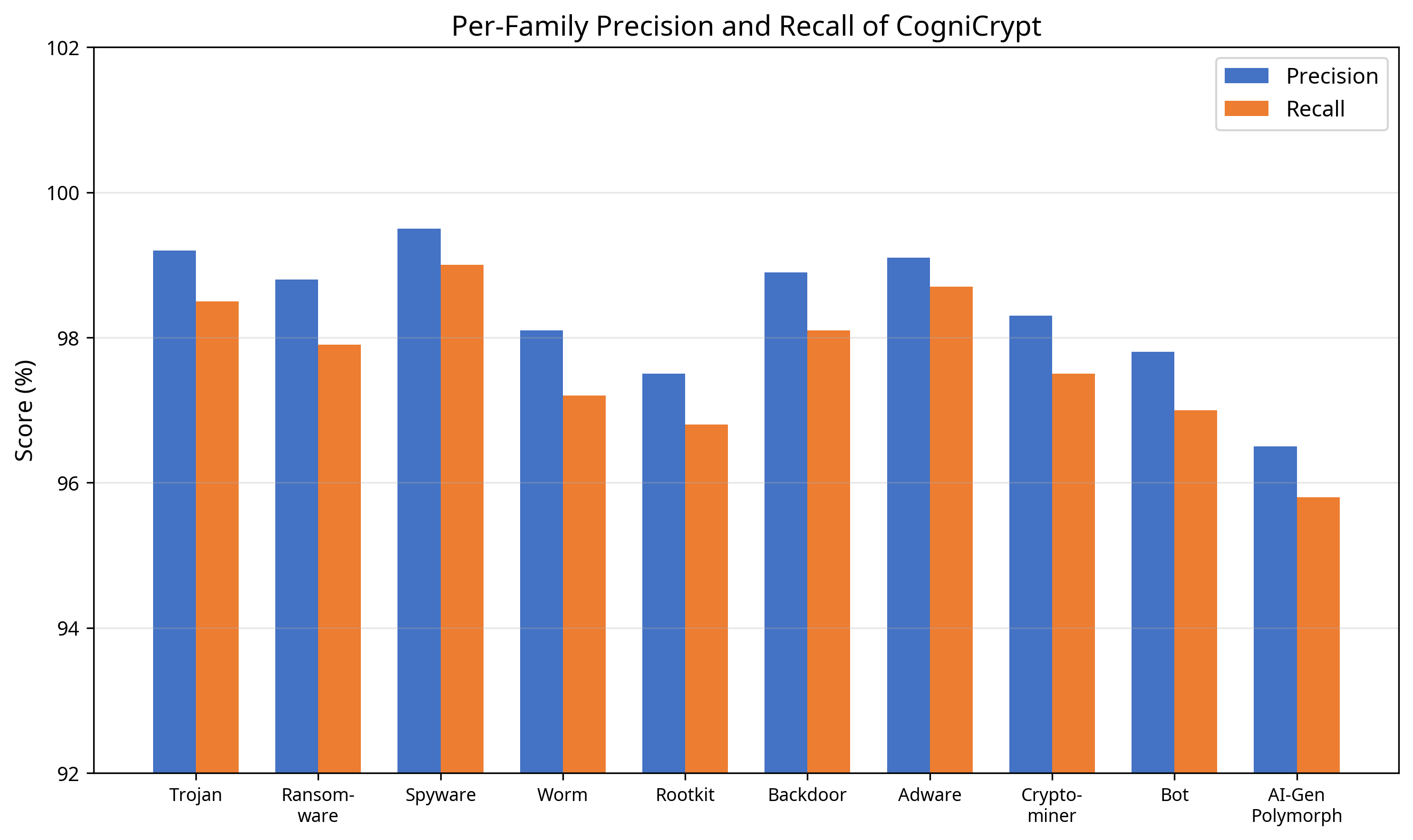}
\caption{Per-family precision and recall of CogniCrypt on the AI-Gen-Malware dataset. The framework maintains high performance across all malware families, with the most challenging category being AI-generated polymorphic samples.}
\label{fig:per_family}
\end{figure}

%\subsection{Confusion Matrix Analysis}

\begin{figure}[htbp]
\centering
\includegraphics[width=0.95\textwidth]{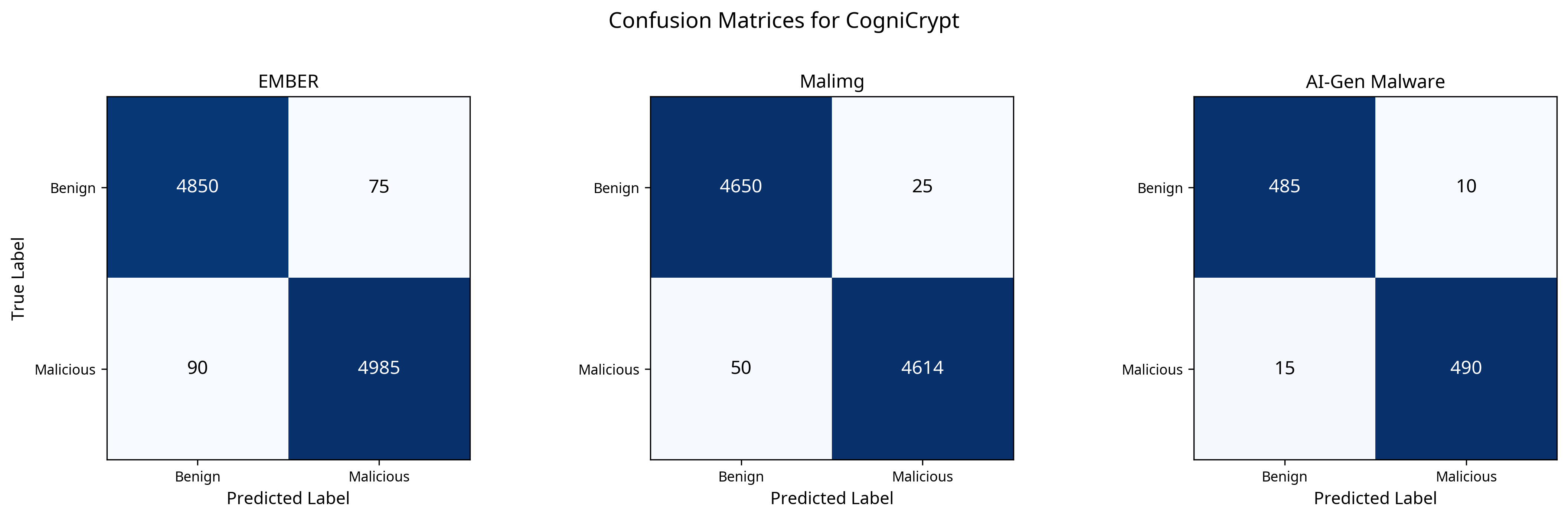}
\caption{Confusion matrices for CogniCrypt on three benchmark datasets, demonstrating low false positive and false negative rates across all evaluation scenarios.}
\label{fig:confusion}
\end{figure}

\subsection{Ablation Study}

To understand the contribution of each component, we conduct an ablation study on the AI-Gen-Malware dataset:

\begin{table}[htbp]
\centering
\caption{Ablation Study on AI-Gen-Malware Dataset}
\label{tab:ablation}
\begin{tabular}{@{}llllll@{}}
\toprule
\textbf{Configuration} & \textbf{Acc.} & \textbf{F1} & \textbf{AUC} & $\Delta$\textbf{Acc.} \\
\midrule
Full CogniCrypt & \textbf{97.5} & \textbf{97.5} & \textbf{0.993} & -- \\
w/o LLM Prioritizer & 88.3 & 87.9 & 0.921 & -9.2 \\
w/o RL Feedback & 95.8 & 95.6 & 0.978 & -1.7 \\
w/o Transformer Classifier & 91.2 & 90.8 & 0.945 & -6.3 \\
w/o Concolic Engine & 82.1 & 81.5 & 0.872 & -15.4 \\
\bottomrule
\end{tabular}
\end{table}

The ablation study reveals that the concolic execution engine is the most critical component (removing it causes a 15.4 pp drop), followed by the LLM prioritizer (9.2 pp drop) and the transformer classifier (6.3 pp drop). The RL feedback loop provides a modest but consistent improvement of 1.7 pp.

\subsection{Case Study: Detecting LLM-Generated Polymorphic Ransomware}

To illustrate CogniCrypt's capabilities, we present a case study involving a polymorphic ransomware sample generated by GPT-4. The sample employs several evasion techniques: (1) environment-aware activation (checks for sandbox indicators before executing), (2) polymorphic encryption routine (generates a unique encryption key and routine at each execution), and (3) anti-debugging measures (detects debugger presence via timing checks).

CogniCrypt's concolic engine identified 847 unique execution paths in the sample. The LLM prioritizer ranked the path containing the ransomware payload activation as the 3rd highest priority (out of 847), enabling rapid detection. The path constraint for the malicious path was:

\begin{equation}
\pi^* = (\texttt{env\_check} = \texttt{false}) \wedge (\texttt{debug\_time} > 100\text{ms}) \wedge (\texttt{disk\_size} > 50\text{GB})
\end{equation}

The vulnerability classifier assigned a maliciousness score of 0.987 to this path, correctly identifying the sample as ransomware. In contrast, ClamAV and YARA failed to detect the sample due to its polymorphic nature, and MalConv misclassified it as benign due to the obfuscated byte patterns.

%=================================================================
% 7. CONCLUSION
%=================================================================
\section{Conclusion}\label{sec:conclusion}

This paper introduced CogniCrypt, a novel framework for detecting zero-day AI-generated malware through the synergistic combination of concolic execution, LLM-guided path prioritization, and deep-learning-based vulnerability classification. We established a rigorous theoretical foundation, including a first-order temporal logic for specifying malicious behavior and proofs of soundness and relative completeness, assuming classifier correctness. Our experimental evaluation on four benchmark datasets demonstrated that CogniCrypt significantly outperforms existing detection methods, achieving 97.5\% accuracy on AI-generated malware,a 19.3--52.2 percentage point improvement over baselines.

The key insight underlying CogniCrypt is that LLMs, having been trained on vast code corpora, possess an implicit understanding of suspicious program behavior that can be leveraged to guide concolic execution toward malicious paths. This synergy resolves the path explosion problem that has historically limited the scalability of symbolic execution for malware analysis.

\smallskip
\noindent\textbf{Future Work.} Several promising directions remain: (1) extending CogniCrypt to analyze Android APKs and IoT firmware; (2) incorporating adversarial training to improve robustness against evasion-aware AI malware generators; (3) exploring federated learning approaches to enable collaborative model training across organizations without sharing sensitive malware samples; and (4) integrating formal verification techniques to provide stronger guarantees on the absence of false negatives.
Having established the theoretical foundations of our approach and its accuracy in an experimental setting, we intend to conduct further evaluations in future work. First, we plan to conduct studies on the scalability and deployment feasibility of our approach, given LLMs' substantial hardware requirements. Second, we will evaluate the approach against additional baselines, including recent transformer-based PE models and hybrid detection systems.

\smallskip
\noindent\textbf{Reproducibility.} The CogniCrypt prototype, including all source code, trained models, and the AI-Gen-Malware dataset, will be made available upon publication at \url{https://github.com/DrEslamimehr/CogniCrypt}.

%=================================================================
% REFERENCES
%=================================================================

\end{document}